\documentclass[a4paper,english,numberwithinsect]{lipics}
\usepackage[utf8]{inputenc}
\usepackage{amsmath}
\usepackage{amssymb}
\usepackage{xcolor}
\usepackage{amsthm}
\usepackage{cleveref}
\usepackage{complexity}
\usepackage{mathabx}

{\theoremstyle{definition}
\newtheorem{rem}[theorem]{Remark}}

\makeatletter{}\newcommand{\egdef}{\stackrel{\mbox{\begin{scriptsize}def\end{scriptsize}}}{=}}
\newcommand{\bigand}{\bigwedge}
\newcommand{\bigor}{\bigvee}
\newcommand{\subword}{\sqsubseteq}
\newcommand{\strictsubword}{\sqsubset}
\newcommand{\supword}{\sqsupseteq}
\newcommand{\xra}{\xrightarrow}
\newcommand{\calA}{\mathcal A}
\newcommand{\calP}{\mathcal P}
\newcommand{\Nat}{\mathbb{N}}
\def\FO{\ComplexityFont{FO}}

\newcommand{\ssubword}{\sqsubset}
\newcommand{\calR}{\mathcal R}
\newcommand{\ssupword}{\sqsupset}

\newcommand{\expr}{\textit{expr}}
\def\TQBF{\ComplexityFont{TQBF}}

\newcommand{\pair}[2]{\begin{bmatrix} #2 \\ #1 \end{bmatrix}}
\newcommand{\VAL}{\textit{V}}
\newcommand{\True}{\textit{true}}
\newcommand{\False}{\textit{false}}

\title{Decidability in the logic of subsequences and
  supersequences\footnote{The first author was partially supported by Tata Consultancy Services and the CEFIPRA Raman-Charpak fellowship. This work was done when he was affiliated to Chennai Mathematical Institute, India and LSV, ENS Cachan, France.}}
\author[1]{Prateek Karandikar}
\author[2]{Philippe Schnoebelen}
\affil[1]{LIAFA, University Paris Diderot, France}
\affil[2]{Laboratoire Specification \& Verification (LSV), Cachan, France}
\authorrunning{P. Karandikar and Ph. Schnoebelen}
\Copyright{Prateek Karandikar and Philippe Schnoebelen}


\begin{document}
\maketitle

\makeatletter{}\begin{abstract}
We consider first-order logics of sequences ordered by the subsequence
ordering, aka sequence embedding. We show that the $\Sigma_2$ theory
is undecidable, answering a question left open by Kuske. Regarding
fragments with a bounded number of variables, we show that the $\FO^2$
theory is decidable while the $\FO^3$ theory is undecidable.
\end{abstract}
 
\makeatletter{}\section{Introduction}
\label{sec-intro}

A subsequence of a (finite) sequence $u=(x_1, \ldots, x_\ell)$ is a
sequence obtained from $u$ by removing any number of elements. For
example, if $u=(a, b, a, b, a)$ then $u'=(b, b ,a)$ is a subsequence
of $u$, a fact we denote with $u'\subword u$. Other examples that work
for any $u$ are $u\subword u$ (remove nothing) and $()\subword u$.

In this paper we consider decidability and complexity questions for
the first-order logic of finite sequences with the subsequence ordering as
the only predicate. The notion of subsequence is certainly a
fundamental one in logic, and it occurs prominently in several areas
of computer science: in pattern matching (of texts, of DNA strings,
etc.), in coding theory, in algorithmics, and in many other areas. We
also note that sequences and their subsequences are a special case of
a more general notion where a family of finite labelled structures
(e.g., trees, or graphs, or ..) are compared via a notion of
embedding. Closer to our own motivations, the automatic verification
of unreliable channel systems and related problems generate many
formulae where the subsequence ordering appears
prominently~\cite{abdulla-forward-lcs,BMOSW-fac2012,HSS-lmcs,KS-msttocs}.

While decision methods for logics of sequences have been considered in
several contexts, the corresponding logics usually do not include the
subsequence predicate:
they rather consider the prefix ordering, and/or membership in a
regular language, and/or functions for taking contiguous subsequences
or computing the length of sequences, see,
e.g.,~\cite{hooimeijer2012,ganesh2013,abdulla2014}. \\

As far as we know, Kuske's article~\cite{kuske2006} is the only one
that specifically considers the decidability of the first-order logic
of the subsequence ordering \textit{per se}. The article also
considers more complex orderings since these decidability questions
first occurred in automated deduction under the name of \emph{ordered
  constraints solving} and they involve rather specific orderings on
terms and strings~\cite{comon90}.

Kuske considers the first-order logic of subsequences over a set of
atoms $A$, denoted $\FO(A^*,\subword)$,
and notes that the undecidability of its $\Sigma_4$ theory
can be seen  by reinterpreting an earlier undecidability result
from~\cite{comon94} for the first-order logic of the
lexicographic path ordering. He then  shows 
 that already the
$\Sigma_3$ theory is undecidable even when $A$ contains only two
elements, and also shows that the
$\Sigma_1$ theory is decidable so that the status of the $\Sigma_2$
theory remains open.

\noindent
\emph{Our contribution.}
In this paper we show that the $\Sigma_2$ theory of the subsequence
ordering is undecidable. On the positive side, we show that the
$\FO^2$ theory is decidable (but $\FO^3$ is not). We also prove some
complexity bounds for the decidable fragments: the $\Sigma_1$ theory
is $\NP$-complete and the $\FO^2$ theory is $\PSPACE$-hard.

\noindent
\emph{Outline of the paper.} The relevant definitions and basic
results are given in section~\ref{sec-basics}.
Section~\ref{sec-S2-undec} develops the reduction that proves
undecidability for the $\Sigma_2$ and $\FO^3$ theories.
Section~\ref{sec-purify} presents a further reduction that proves
undecidability for the $\Sigma_2$ theory even when constants are not
allowed in the formulae. Then section~\ref{sec-FO2} shows decidability
for the two-variable fragment $\FO^2$.

Since our constructions heavily rely on concepts and results from
formal language theory, we shall from now on speak of ``words'', and
``letters'' (from an ``alphabet'') rather than sequences and atoms.
Note however that the logic $\FO(A^*,\subword)$ is defined for any
kind of set $A$.

\makeatletter{}\section{Basic notions}
\label{sec-basics}

Let $A=\{a_1,a_2,\ldots\}$ be a set called \emph{alphabet}, whose
elements are called \emph{letters}. 
In this paper we only consider finite alphabets for ease of exposition
but without any real loss of generality.
A \emph{word} is a finite sequence
of letters like $a a c$ and we use $u$, $v$, \ldots, to denote
words, and $A^*$ to denote the set of all words over $A$.
Concatenation of word is written multiplicatively, and $\epsilon$
denotes the empty word. We also use regular expressions like
$(a b+c)^*$ to denote regular languages (i.e., subsets of $A^*$). The
length of a word $u$ is denoted $|u|$ and, for $a\in A$, we let
$|u|_a$ denote the number of occurrences of $a$ in $u$.

We say that a word $u$ is a \emph{subword} (i.e., a subsequence) of
$v$, written $u\subword v$, when $u$ is some $a_1\cdots a_n$ and $v$
can be written under the form $v_0a_1 v_1 \cdots a_n v_n$ for some
$v_0,v_1,\ldots,v_n\in A^*$. We say a word $u$ is a \emph{factor} of a
word $v$ if there exist words $v_1$ and $v_2$ such that $v = v_1 u
v_2$. For $B \subseteq A$, and $w \in A^*$, we define the projection
of $w$ onto $B$, denoted as $\pi_B(w)$, as the subword of $w$ obtained
by removing all letters in $A \setminus B$. For example,
$\pi_{\{a,b\}}(abcaccbbc) = ababb$.

We assume familiarity with basic notions of first-order logic as
exposed in, e.g.,~\cite{harrison2009}: bound and free occurrences of
variables, etc.

In particular, for $n\in\Nat$, the fragment $\FO^n$ consists of all
formulae that only use at most $n$ distinct variables (these can have
multiple occurrences inside the formula).

The fragments $\Sigma_n$ and $\Pi_n$ of $\FO(A^*,\subword)$
are defined inductively as follows:
\begin{itemize}
\item
an atomic formula is in $\Sigma_n$ and $\Pi_n$ for all $n\in\Nat$;
\item
a negated formula $\neg\phi$ is in $\Sigma_n$ iff $\phi$ is in
$\Pi_n$, it is in $\Pi_n$ iff $\phi$ is in $\Sigma_n$;
\item
a conjunction $\phi\land\phi'$ is in $\Sigma_n$ (resp., in $\Pi_n$)
iff both $\phi$ and $\phi'$ are;
\item
For $n>0$, an existentially quantified $\exists x\phi$ is in $\Sigma_n$ iff
$\phi$ is, it is in $\Pi_{n}$ iff $\phi$ is in
$\Sigma_{n-1}$;
\item
For $n>0$, a universally quantified $\forall x\phi$ is in $\Pi_n$ iff $\phi$ is,
it is in $\Sigma_{n}$ iff $\phi$ is in $\Pi_{n-1}$.
\end{itemize}
Note that we do not require formulae to be in prenex normal form when
defining the $\Sigma_n$ and $\Pi_n$ fragments: for example the formula
 $ \forall x\:\exists y (
x \subword y \land \exists x\: \neg (x \subword y)) $ is
simultaneously in $\Pi_2$ and $\FO^2$.

In this article we consider three versions of $\FO(A^*,\subword)$, the
first-order logic of subsequences over $A$:
\begin{description}
\item[The pure logic:] the signature consists of only one predicate
  symbol, ``$\subword$'', denoting the subword relation. One also uses
  a countable set $X=\{x,x',y,z,\ldots\}$ of variables ranging over
  words in $A^*$ and the usual logical symbols.

Note that there is no way in the pure logic to refer to specific
elements of $A$ in the logic. However, whether a formula $\phi$ is
true, denoted $\models_{A^*}\phi$, may depend on $A$ (in fact, its
cardinality). For example, the closed formula
\[
\forall x, y (x\subword y \lor y\subword x)
\:,
\]
stating that $\subword$ is a total ordering, is true if, and only if,
$A$ contains at most one letter.

\item[The basic logic:] extends the pure logic by adding all words
  $u\in A^*$ as constant symbols (denoting themselves). For example,
  assuming $A$ contains $a$, $b$ and $c$, one can write the following sentence:
\[
\exists x (ab\subword x\land bc\subword x\land abc\not\subword x)
\]
which is true, as witnessed by the valuation $x\mapsto b c a b$.

\item[The extended logic:] further allows all regular expressions as
  unary predicates (with the expected semantics). For these predicates
  we adopt a more natural notation, writing e.g.\ $x\in \expr$ rather
  than $P_\expr(x)$. For example, the extended logic allows writing
\[
\forall x\bigl( \bigl[\exists y(y\in (ab)^*\land x\subword y)\bigr] \Leftrightarrow x\in (a+b)^*\bigr)
\]
which states that the regular language $(a+b)^*$ is the downward closure of $(ab)^*$,
i.e., the set of all subwords of its words.
\end{description}
When writing formulae we freely use abbreviations like $x
\strictsubword y$ for $x \subword y \land \neg(y \subword x)$ and $x
\supword y$ for $y \subword x$. Note that equality can be defined as
an abbreviation since $x\subword y\land y\subword x$ is equivalent to
$x=y$. Finally, we use negated symbols as in $x\not\subword y$ or
$x\not\in (ab)^*$ with obvious meaning.

When we write $\FO(A^*,\subword)$ without any qualification we
refer by default to the basic logic.
The pure logic is apparently a very restricted logic, where one may
hardly express more than generic properties of the subword ordering
like saying that $(A^*,\subword)$ is a total ordering, or is a
lattice. However, Theorem~\ref{thm-undec} below shows that the pure
logic is quite expressive.

We conclude this expository section with
\begin{theorem}
\label{thm-S1-NP}
The truth problem for the $\Sigma_1$ fragment of $\FO(A^*,\subword)$
is $\NP$-complete even when restricting to a fixed alphabet.
\end{theorem}
\begin{proof}[Proof sketch]
The upper bound follows from the decidability proof
in~\cite{kuske2006} since it is proved there that a satisfiable
quantifier-free formula $\phi(x_1,\ldots,x_n)$ can be satisfied with
words of size in $O(n)$ assigned to the $x_i$'s. Guessing linear-sized
witnesses $u_1,\ldots,u_n$ and checking that
$\models_{A^*}\phi(u_1,\ldots,u_n)$ can be done in $\NP$.

For the lower bound, we reduce from boolean satisfiability. Consider a
boolean formula $\phi(x_1,\ldots,x_n)$ over $n$ boolean variables. We
reduce it to an $\FO(A^*,\subword)$ formula in the $\Sigma_1$ fragment
\[
\psi \equiv \exists z, x_1, \ldots, x_n (\phi')
\]
where $\phi'$ is obtained from $\phi$ by replacing each occurrence of
$x_i$ with $x_i \subword z$ (hence replacing $\neg x_i$ with
$x_i\not\subword z$). Then, for any alphabet $A$  with at least one
letter,  $\phi$ is satisfiable if and only if
$\models_{A^*} \psi$.
\end{proof}

\makeatletter{}\section{Undecidability for $\Sigma_2$}
\label{sec-S2-undec}

We are interested in solving the \emph{truth problem}. This asks,
given an alphabet $A$ and a sentence $\phi\in\FO(A^*,\subword)$,
whether $\phi$ is true in the structure $(A^*,\subword)$, written
$\models_{A^*}\phi$. Restricted versions of the truth problems are
obtained for example by fixing $A$ (we then speak of the truth problem
\emph{over $A$}) and/or by restricting to a fragment of the logic.

This section is devoted to proving the following main result.

\begin{theorem}[Undecidability]
\label{thm-undec}
The truth problem for $\FO(A^*,\subword)$ is undecidable even when
restricted to formulae in the $\Sigma_2\cap\FO^3$ fragment of the
basic logic. 
\end{theorem}

This is done by encoding Post's Correspondence Problem in
$\FO(A^*,\subword)$. The reduction is described in several stages.

\subsection{Expressing simple properties}

We start with a list of increasingly complex properties and show how
to express them in the basic $\FO(A^*,\subword)$ logic.
We keep track of what fragment is used, with regards to both the
number of distinct variables, and the quantifier alternation depth.

Note that when we claim that a property with $m$ free variables can be
expressed in $\FO^n$ (necessarily $n \geq m$), we mean that the
formula only uses at most $n$ variables \emph{including the $m$ free
variables}.

We let  $A=\{a_1,...,a_\ell\}$ denote an arbitrary alphabet, use $B$ to denote
subsets of $A$, and $a,b,...$ to denote arbitrary letters from $A$.

\begin{description}
\item[P1. `` $x \in B^*$ '' can be expressed in $\Sigma_0\cap\FO^1$:] using
\[
\bigand_{a \in A \setminus B} a \not \subword x
\:.
\]

\item[P2. `` $\pi_B(y) \subword x$ '' can be expressed in $\Pi_1\cap\FO^3$:]
building on P1, we use
\[
\forall z  \bigl( ( z \subword y \land z \in B^* ) \implies z \subword
x\bigr)
\:,
\]
noting that $\pi_B(y) \subword x$ is equivalent to $\pi_B(y) \subword \pi_{B}(x)$.

\item[P3. `` $x = \pi_B(y)$ '' can be expressed in $\Pi_1\cap\FO^3$:]
building on P1, P2, and using
\[
\pi_B(y) \subword x \land x \subword y \land x \in B^*
\:.
\]

\item[P4. `` $\pi_B(x) = \pi_B(y)$ '' can be expressed in
$\Pi_1\cap\FO^3$:]
building on P2, and using
\[
\pi_B(y) \subword x \land \pi_B(x) \subword y
\:.
\]

\item[P5. `` $x\in aA^*$ '', i.e., `` $x$ starts with $a$ '', can be expressed in $\Sigma_2\cap\FO^3$:] building \linebreak on P1, and using
\[
\exists z \Bigl( a \subword z \land\bigl[\bigand_{\!\!b \in A \setminus \{a\}\!\!} ba \not
  \subword z\bigr] \land z \subword x \land \pi_{A \setminus \{a\}}(x) \subword z
\Bigr)
\:.
\]
Here the first two conjuncts require that $z$ contains an occurrence of $a$
and cannot start with another letter. The last two conjuncts require
that $z$ is a subword of $x$ which has at least all the occurrences in $x$ of all letters other than $a$.

Clearly, the mirror property ``$x\in A^*a$'' can be expressed in
$\Sigma_2\cap\FO^3$ too.

\item[P6. `` $x\not\in A^*aaA^*$ '' can be expressed in $\Sigma_2\cap\FO^3$:] building on
  P3, and using
\[
\exists y \Bigl( y = \pi_{A \setminus \{a\}}(x) \land \forall z\bigl[
(aa \subword z \land y \subword z \land z \subword x) \implies \bigvee_{\!\!b \in A\setminus\{a\}\!\!} aba \subword z \bigr]\Bigr)
\:.
\]
Note that this is equivalent to `` $x$ does not have $aa$ as
  a factor''. Here $z\subword x$ implies that any two occurrences of $a$ in $z$ must
come from $x$. Furthermore, if these are not contiguous in $x$ they
cannot be contiguous in $z$ in view of
$y=\pi_{A\setminus\{a\}}(x)\subword z$.
\end{description}
\begin{rem}
\label{rem-reuse}
Note that the ``$y = \pi_{A \setminus \{a\}}(x)$'' subformula in P6 uses one variable apart from $y$ and $x$. We use the same variable name $z$ that is used later in the formula, so that the formula is in $\FO^3$. We similarly reuse variable names whenever possible in later formulae.
\end{rem}
\begin{description}
\item[P7. `` $x\not\in A^*BBA^*$ ''
  can be expressed in $\Sigma_2\cap\FO^3$:] as in P6 with
\[
\exists y \Bigl( y = \pi_{A \setminus B}(x) \land \forall z\!\!\bigand_{\!\!a,a' \in B\!\!}\bigl[
  ( aa' \subword z \land y \subword z
 \land z \subword x) \implies
\bigor_{\!\!b \in A \setminus B\!\!} aba' \subword z \bigr]\Bigr)
\:.
\]
Note that this is equivalent to `` $x$ has no factor in $BB$ ''.

\item[P8. `` $|\pi_B(x)| = 2$ '' can
  be expressed in $\Sigma_0\cap\FO^1$:] using
\[
\Bigl(\bigor_{\!\!a,a' \in B\!\!} a a' \subword x \Bigr) \;\land \bigand_{\!\!\!\!a,a',a'' \in B\!\!\!\!} a a' a'' \not \subword x
\:.
\]
\end{description}

\subsection{Expressing regular properties}

Building on the previous formulae, our next step is to show how any
regular property can be expressed in the basic logic by using an
enlarged alphabet.

\begin{lemma}
\label{lem-express-regular}
For any regular $L\subseteq A^*$ there is an extended alphabet
$A'\supseteq A$ and a formula $\phi_L(x)$ in $\Sigma_2\cap\FO^3$ over
$A'$ such that for all $u \in A'^*$, $u\in L$ if and only if $\models_{A'{}^*}\phi_L(u)$.
\end{lemma}
\begin{proof}
Let $\calA=(Q, A, \delta, I, F)$ be a NFA recognising $L$ so that
$u\in L$ iff $\calA$ has an accepting run on input $u$. We define
$\phi_L(x)$ so that it states the existence of such a run, i.e., we
put $\phi_L(x)\equiv \exists y\:\psi_\calA(x,y)$ where
$\psi_\calA(x,y)$ expresses that ``$y$ is an accepting run of $\calA$
over $x$.'' \\

Let $A'\egdef A \cup Q$, assuming w.l.o.g.\  that $A$ and $Q$ are
disjoint. A run $q_0 \xra{a_1} q_1 \xra{a_2} \ldots \xra{a_n} q_n$ of
$\calA$ can be seen as a word $q_0a_1q_1a_2\ldots a_n q_n$ in
$A'{}^*$.
We now define $\psi_\calA(x,y)$ as the
conjunction $\psi_1(x,y)\land\psi_2(x,y)$, with
\begin{align*}
\psi_1&\equiv\;
\begin{array}{>{\displaystyle}l}
\;\;\;\;  (y\text{ has no factor from } AA)\ \land \ (y\text{ has no factor from } QQ)
\\[.3em]
\land\ \bigl( \bigor_{q \in I} y\text{ begins with } q\bigr)\ \land\
\bigl( \bigor_{q \in F} y\text{ ends with } q\bigr)
\ \land \ (\pi_A(y)=x )\:,
\end{array}
 \\[1em]
\psi_2&\equiv\;
\forall z
\Biggl(
\begin{array}{>{\displaystyle}l}
(x \subword z \land z \subword y \land z\text{ has exactly two occurrences of
  letters from } Q)\\[.3em]
\;\;\;\implies
\Bigl(\bigor_{q,q'\in Q}\bigor_{a,a'\in A} q a a' q' \subword z
\:\lor\: \bigor_{(q, a, q') \in \delta} q a q' \subword z \Bigr)
\end{array}\Biggr)
\:.
\end{align*}
Here $\psi_1$ reuses simple properties from the previous subsection
and states that $y$ is a word alternating between $Q$ (states of
$\calA$) and $A$ (proper letters), starting with an initial state of
$\calA$ and ending with an accepting state, hence has the required
form $q_0 a_1 \ldots a_n q_n$. Furthermore,
$\pi_A(y)=x$ ensures that $y$ has the form of an accepting run over
$x$. Note that it also ensures $x \in A^*$.

With $\psi_2$, one further ensures that the above $y$ respects the
transition table of $\calA$, i.e., that $(q_{i-1},a_i,q_i)\in\delta$
for $i=1,\ldots,n$. Indeed, assume $z\in A'{}^*$ satisfies $x\subword
z\subword y$ and contains two occurrences from $Q$. Thus $z$ is
$a_1\ldots a_{i}q_ia_{i+1}a_{i+2}\ldots a_{j}q_ja_{j+1}a_{j+2}\ldots
a_n$ for some $1\leq i<j\leq n$. If now $j>i+1$ then $z$ contains
$q_ia_{i+1}a_{i+2}q_j$ as a subword and the disjunction after the
implication is fulfilled. However, if $j=i+1$, the only way to fulfil
the disjunction is to have $(q_{j-1},a_j,q_j)\in\delta$.

Finally, $\psi_\calA(x,y)$ exactly states that $y$ is an accepting run
for $x$ and $\models_{A'{}^*}\phi_L(u)$ holds iff $u\in L$. One easily
checks that $\psi_1$ is in $\Sigma_2 \cap \FO^3$, $\psi_2$ is in $\Pi_1 \cap \FO^3$, so that
$\psi_\calA$ and $\phi_L$ are in $\Sigma_2 \cap \FO^3$. We reuse variables
wherever possible to ensure that only three variables are used (see remark
\ref{rem-reuse}). For example, the implementation of ``$y$
has no factor from $QQ$'' from P7 needs two other variables, and here
we use $x$ and $z$ for it.
\end{proof}

\subsection{Encoding Post's Correspondence Problem}
\label{ssec-PCP-encoding}

It is now easy to reduce Post's Correspondence Problem to the truth
problem for the basic $\FO(A^*,\subword)$ logic.

Suppose we have a PCP instance $\calP$ consisting of pairs
$(u_1,v_1),\ldots (u_n,v_n)$ over the alphabet $\Gamma$. We let $N =
\{1, \ldots, n\}$, consider the alphabet $A\egdef \Gamma\cup N$, and
define
\begin{equation}
\label{eq-phiP}
\phi_\calP \equiv \exists x, x'
\biggl( \begin{array}{l}
        \;\;\;\; x \in (1u_1 + \cdots + nu_n)^+ \ \land \ x' \in (1v_1 + \cdots + nv_n)^+
\\[.3em]
        \land \ \pi_N(x) = \pi_N(x') \land \pi_\Gamma(x) = \pi_\Gamma(x')
\end{array} \biggr)
\:.
\end{equation}
Clearly, $\phi_\calP$ is true iff the PCP instance has a solution.

It remains to check that $\phi_\calP$ is indeed a formula in the
$\Sigma_2$ fragment: this relies on
Lemma~\ref{lem-express-regular} for expressing membership in two
regular languages, and the P4 properties for ensuring that $x$ and
$x'$ contain the same indexes from $N$ and the same letters from
$\Gamma$. Finally, we note that $\phi_\calP$ is also a $\FO^3$ formula.

\makeatletter{}\section{Undecidability for the pure logic}
\label{sec-purify}

In this section we give a stronger version of the undecidability
for the $\Sigma_2$ fragment.

\begin{theorem}[Undecidability for the pure logic]
\label{thm-undec-pure}
The truth problem for $\FO(A^*,\subword)$ is undecidable even when
restricted to formulae in the $\Sigma_2$ fragment of the
pure logic.
\end{theorem}

The proof is by constructing a $\Sigma_2$ formula $\psi(x_1,\ldots)$
in the pure logic that defines all the letters and constant words
we need to reuse the reduction from the previous section.

Kuske solves the problem in the special case of a
formula using only $\{\epsilon,a,b,ab,ba,aa,\linebreak[0] bb, \linebreak[0] aba,bab\}$ as
constants~\cite{kuske2006}.
We provide a more generic construction whereby all words (up to a
fixed length) can be defined in a single $\Sigma_2$ formula.
One inherent difficulty is that it is impossible to properly define
constant words in the pure logic. Of course, with the pure
logic one can only define properties up to a bijective renaming of
the letters, so $\psi(x_1,\ldots)$ will only define letters and words
up to renaming. But a more serious problem is that we can only define
properties \emph{invariant by mirroring} as we now explain.

For a word $u=a_1 a_2\ldots a_\ell$, we let
$\widetilde{u}$ denote its mirror image $a_\ell\ldots a_2 a_1$.
\begin{lemma}[Invariance by mirrorring]
\label{lem-mirroring}
If $\psi(x_1,\ldots,x_n)$ is a formula in the pure logic and
$u_1,\ldots,u_n$ are words in $A^*$, then
$\models_{A^*}\phi(u_1,\ldots,u_n)$ if, and only if,
$\models_{A^*}\phi(\widetilde{u_1},\ldots,\widetilde{u_n})$.
\end{lemma}
\begin{proof}[Proof Sketch]
By structural induction on $\phi$, noting that the only atomic
formulae in the pure logic have the form $x\subword y$, and that
$u\subword v$ iff $\widetilde{u}\subword\widetilde{v}$ for any $u,v\in
A^*$.
\end{proof}

\subsection{Defining letters and short constant words}

We now define $\psi(x_1,\ldots)$. In our construction $\psi$ has the
form $\psi_1\land \psi_2\land\cdots \land\psi_{13}$ and features a large number
of free variables. We describe the construction in several stages,
explaining what valuation of its free variables can make $\psi$ true.
We start with
\begin{align}
       & \textstyle{\forall y(z\subword y)}
\tag{$\psi_1$}
  \\
\land\;&\textstyle{\bigand_{1\leq i\not=j\leq n} x_i\not\subword x_j}
\tag{$\psi_2$}
\\
\land\;&  \textstyle{\bigand_{i=1}^n
 \forall y [y\subword x_i\implies (x_i\subword y\lor y\subword z)]}
\tag{$\psi_3$}
\end{align}
Here $\psi_1$ implies $z=\epsilon$, then $\psi_2$ implies
$x_i\neq\epsilon$ so that $\psi_3$ requires that each $x_i$ is a
single letter and furthermore $x_1,\ldots,x_n$ must be different
letters as required by $\psi_2$.

We continue with:
\begin{align}
\land\;&
  \textstyle{\bigand_{i=1}^n
\bigl( x_i\subword x^2_i\land x^2_i\not\subword x_i\land \forall y[y \subword x^2_i\implies (y\subword x_i\lor x^2_i\subword y)] \bigr)}
\tag{$\psi_4$}
\end{align}
Note that $n$ new free variables, $x^2_1,\ldots,x^2_n$  are
involved. First $\psi_4$ requires that any $x^2_i$ has at least two
letters (it must contain $x_i$ strictly). But it also requires that
any subword of $x^2_i$ is $\epsilon$ or $x_i$ or $x^2_i$, thus $x^2_i$
has length $2$ and
 can only be $x_ix_i$.

In the same style we introduce new free variables $x^3_1,\ldots,x^3_n$
and $x^4_1,\ldots,x^4_n$ and require that $x^3_i$ equals $x_i x_i
x_i$, and that $x^4_i$ equals $x_i x_i x_i x_i$ with:
\begin{align}
\land\;&
\textstyle{\bigand_{i=1}^n
\bigl( x^2_i\subword x^3_i\land x^3_i\not\subword x^2_i\land \forall y[y \subword x^3_i\implies (y\subword x^2_i\lor x^3_i\subword y)] \bigr)}
\tag{$\psi_5$}
\\
\land\;&
\textstyle{\bigand_{i=1}^n
\bigl( x^3_i\subword x^4_i\land x^4_i\not\subword x^3_i\land \forall y[y \subword x^4_i\implies (y\subword x^3_i\lor x^4_i\subword y)] \bigr)}
\tag{$\psi_6$}
\end{align}

We introduce new free variables $\{y_{i,j}\}_{1\leq i\neq j\leq n}$ and conjuncts:
\begin{align}
\land\;&
\textstyle{\bigand_{1\leq i\neq j\leq n}
\forall y\:(y\subword y_{i,j}\implies y\subword z\lor x_i\subword y\lor
x_j\subword y)}
\tag{$\psi_8$}
\\
\land\;&
\textstyle{\bigand_{1\leq i\neq j\leq n} \bigl(
x_i\subword y_{i,j}\land
x_j\subword y_{i,j}\land
x^2_i\not\subword y_{i,j}\land
x^2_j\not\subword y_{i,j})}
\tag{$\psi_9$}
\\
\land\;&
\textstyle{\bigand_{1\leq i\neq j\leq n} \bigl(y_{i,j}\not\subword y_{j,i}\bigr)}
\tag{$\psi_{10}$}
\end{align}
Here $\psi_8$ requires that any $y_{i,j}$ only contains letters among
$x_i$ and $x_j$, and $\psi_9$ requires that it contains exactly one
occurrence of $x_i$ and one of $x_j$. So that $y_{i,j}$ is either
$x_ix_j$ or $x_jx_i$.  With $\psi_{10}$ we require that  $y_{j,i}$
is, among $x_ix_j$ and $x_jx_i$,  the word not assigned to $y_{i,j}$.

Now, in view of Lemma~\ref{lem-mirroring}, it is impossible to fix
e.g.\ $y_{i,j}=x_ix_j$. However we can force all $y_{i,j}$ to have
``the same orientation''. Let $i,j,k$ be three different indexes in
$\{1,\ldots,n\}$ and consider the following formula
\begin{align*}
\xi_{i,j,k}\equiv
\exists t
\left[\begin{array}{rlcl}
&\forall y(y\subword t\implies y\subword z\lor x_i\subword y\lor
  x_j\subword y\lor x_k\subword y)
&\;\;\;\;&(\xi_1)
\\
\land&
x^2_i\subword t \land x^3_i\not\subword t\land x_j\subword t\land x^2_j\not\subword t\land x_k\subword t\land x^2_k\not\subword t
&&(\xi_2)
\\
\land&
y_{i,j}\subword t\land y_{j,i}\subword t
\land y_{i,k}\subword t\land y_{k,i}\not\subword t
\land y_{j,k}\subword t\land y_{k,j}\not\subword t
&&(\xi_3)
\end{array}
\!\right]
\end{align*}
We claim that, in conjunction with the earlier $\psi$-conjuncts,
$\xi_1\land\xi_2\land \xi_3$ requires $t=x_ix_jx_ix_k$ or
$t=x_kx_ix_jx_i$: indeed by $\xi_1$, $t$ only contains letters among
$\{x_i,x_j,x_k\}$, then by $\xi_2$, $t$ contains exactly 2 occurrences
of $x_i$ and exactly one occurrence each of $x_j$ and $x_k$, then by
$\xi_3$, $t$ has $x_ix_j$ and $x_jx_i$ as subwords, so the single
occurrence of $x_j$ is between the two occurrences of $x_i$ and, by
$\xi_3$ again, the occurrence of $x_k$ is outside the two $x_i$
occurrences.
Finally, satisfying $\xi_{i,j,k}$ requires $y_{i,k}$ and $y_{j,k}$ to
have the same orientation.
\\

We continue the construction of $\psi$ with:
\begin{align}
\land\;&
  \textstyle{\bigand_{1\leq i\neq j\leq n} \bigand_{k\not\in\{i,j\}}\xi_{i,j,k}}
\tag{$\psi_{11}$}
\end{align}
As just explained, this will force all $y_{i,j}$'s to have the same
orientation, i.e., any satisfying assignment will have
$y_{i,j}=x_ix_j$ for all $i,j$, or $y_{i,j}=x_jx_i$ for all $i,j$. \\

\subsection{Defining long constant words}

Once we have defined all words of length $2$ (up to mirroring) over
the alphabet $\{x_1,\ldots,x_n\}$ (up to renaming), it is easier to
systematically define all words of length $3$, $4$, etc. Actually, we
only use constant words of length at most $4$ for the formula $\phi_\calP$ from
section~\ref{sec-S2-undec}.

The general strategy relies on a technical lemma we now explain. For
$n\in\Nat$ we say that two words $u$ and $v$ are $n$-equivalent, written
$u \sim_n v$, if $u$ and $v$ have the same set of subwords of length
up to $n$. Thus $\sim_n$ is the piecewise-testability congruence
introduced by Simon, see~\cite{simon75,sakarovitch83}.
\begin{lemma}[See appendix]
\label{lem-simnnp1}
Let $n \geq 2$, and let $u$ and $v$ be words of length $n+1$ with $u \neq v$. Then $u \not \sim_n v$.
\end{lemma}

We can thus introduce new variables $y_{i,j,k}$ and $y_{i,j,k,m}$ for
all $i,j,k,m\in\{1,\ldots,n\}$ (allowing repetitions of indexes) and
require $y_{i,j,k}=x_ix_jx_k$ and $y_{i,j,k,m}=x_ix_jx_kx_m$, up to
mirroring but with the same orientation for all the
$y_{i_1,\ldots,i_\ell}$'s. Then we complete the construction of $\psi$
with the following conjuncts:
\begin{align}
\land\;&
\textstyle{\bigand_{1\leq i,j,k\leq n} \textit{``formula defining $y_{i,j,k}$''}}
\tag{$\psi_{12}$}
\\
\land\;&
  \textstyle{\bigand_{1\leq i,j,k,m\leq n} \textit{``formula defining
      $y_{i,j,k,m}$''}}
\:.
\tag{$\psi_{13}$}
\end{align}

In order to require that, for example, $y_{1,5,2}=x_1x_5x_2$, it is
enough to:
\begin{itemize}
\item enumerate all words of length upto $2$, and for each say whether it is or is not a subword of $y_{1,5,2}$ ($y_{1,5} \subword y_{1,5,2} \land x_1^2 \not \subword y_{1,5,2} \land \ldots$),
\item and require that $y_{1,5,2}$ has length $3$, by saying that every subword of $y_{1,5,2}$ is itself or is one of the words of length upto $2$, and that $y_{1,5,2}$ is distinct from all these words.
\end{itemize}
The correctness of the construction is guaranteed by Lemma~\ref{lem-simnnp1}.

Once all 3-letter words have been defined, we can use them to define
4-letter words (and if needed, 5-letter words, and so on) simlarly, with correctness following from Lemma~\ref{lem-simnnp1}.

Finally, we let $\phi'_\calP$ be obtained from the formula $\phi_\calP$ ---see
Eq.~\eqref{eq-phiP} page~\pageref{eq-phiP}--- by replacing  every constant
letter $a_i\in A$  by the variable $x_i$, and every
constant word $a_{i_1} \ldots a_{i_\ell}\in A^*$  by the
variable $y_{i_1,\ldots,i_\ell}$ (we use $z$ for the
constant word $\epsilon$, and $x_i^2$ for the constant word $x_ix_i$).

Now we define $\psi_\calP$ with
\[
\psi_\calP\equiv \exists Z\:(\psi_1\land \cdots\land\psi_{13}\land \phi'_\calP)
\]
where $Z=\{z,x_1,\ldots,x_n,x'_1,\ldots,x'_n, x_1^2, x_1^3,x_1^4,\ldots, y_{1,1},\ldots,
y_{i_1,\ldots,i_\ell},\ldots\}$ collects all the free variables we
used in $\psi_1\land\cdots\land\psi_{13}$.

Noting that each $\psi_i$ as well as $\phi'_\calP$ is a $\Sigma_2$ formula, we get that the resulting $\psi_\calP$ is a $\Sigma_2$ formula in the pure logic
that is true in $(A^*,\subword)$ iff the PCP instance $\calP$ is
positive.
This concludes the proof of Theorem~\ref{thm-undec-pure}.

\subsection{Undecidability for a fixed alphabet}

The above Theorem~\ref{thm-undec-pure} applies to the truth problem
\emph{for unbounded alphabet}, i.e., where we ask whether
$\models_{A^*}\phi$ for given $A$ and $\phi$. In this proof, the
alphabet $A$ depends on the PCP instance $\calP$ since it includes
symbols for the states of the regular automata that define the
languages $(1u_1+\cdots+n u_n)^+$ and $(1v_1+\cdots+n v_n)^+$ in
Eq.~\eqref{eq-phiP}, and further includes symbols in
$N=\{1,\ldots,n\}$.

It is possible to further show undecidability of the $\Sigma_2$
fragment even \emph{for a fixed alphabet $A$} as we now explain.
For this we consider a variant of Post's Correspondence Problem:
\begin{definition}
The \emph{variant PCP problem} asks, given an alphabet
$\Gamma$, pairs $(u_1,v_1),\ldots,\linebreak[0] (u_n, v_n)$ over $\Gamma$, and an
extra word $w\in\Gamma^*$, whether there exists a sequence
$i_1,\ldots,i_\ell$ over $\{1,\ldots,n\}$ such that
$w\, u_{i_1}\ldots u_{i_\ell} = v_{i_1}\ldots v_{i_\ell}$.
\end{definition}
\begin{lemma}
There is a fixed $\Gamma$ and a fixed sequence of pairs over $\Gamma$
for which the variant PCP problem (with only $w$ as input) is
undecidable.
\end{lemma}
\begin{proof}[Proof Sketch]
One adapts the standard undecidability proof for PCP. Instead of
reducing from the question whether a given TM halts, one reduces from
the question whether a fixed TM accepts a given input. Note that in the
case of a universal TM, the problem is undecidable. Fixing the TM will
lead to a fixed sequence of pairs $(u_1,v_1),\ldots,(u_n,v_n)$, and
the input of the TM will provide the $w$ parameter of the problem.
\end{proof}

\begin{theorem}[Undecidability for fixed alphabet]
There exists a fixed alphabet $A$ such that the truth problem for the pure logic
$\FO(A^*,\subword)$ is undecidable even when restricted to formulae in
$\Sigma_2$.
\end{theorem}
\begin{proof}[Proof Sketch]
We adapt the proof of Theorems~\ref{thm-undec}
and~\ref{thm-undec-pure} by reducing from the variant PCP problem with
fixed $\Gamma$ and sequence of pairs. The encoding formula can be
\begin{equation}
\label{eq-phiP'}
 \equiv \exists x, x'
\biggl( \begin{array}{l}
        \;\;\;\; x \in \Gamma'{}^*\cdot(1u_1 + \cdots + nu_n)^+ \ \land \ x' \in \rho(1v_1 + \cdots + nv_n)^+
\\[.3em]
        \land\ \pi_{\Gamma'}(x)=\hat{w}\ \land \ \pi_N(x) = \pi_N(x') \land \pi_{\Gamma\cup\Gamma'}(x) = \pi_{\Gamma\cup\Gamma'}(x')
\end{array} \biggr)
\end{equation}
to be compared with Eq.~\eqref{eq-phiP}. Here we use
$\Gamma'=\{\hat{a},\hat{b},..\}$, a renamed copy of
$\Gamma=\{a,b,..\}$, to be able to extract the $w$ prefix in $x$.
The word $\hat{w}$ is simply $w$ from the variant PCP instance
with all letters from $\Gamma$ replaced by corresponding letters
from $\Gamma'$. We
then need to extend the language $(1v_1+\cdots + nv_n)$ for $x'$ so
that letters from $\Gamma'$ can be used in place of the corresponding
letters from $\Gamma$. This is done by applying a simple transduction
$\rho\egdef \Bigl(\bigcup_{a\in\Gamma} \bigl[{a\atop
    a}\bigr]\cup\bigl[{\hat{a}\atop a}\bigr]\Bigr)^*$.

In the end, we only use two fixed regular languages, and thus a fixed
alphabet $A$. Note however that encoding the input $w$ will require
using constant words of unbounded lengths. Here we rely on the fact
that our reduction from basic to pure logic can define constant words
of arbitrary length in the $\Sigma_2$ fragment.
\end{proof}

\makeatletter{}\section{Decidability for the $\FO^2$ fragment}
\label{sec-FO2}

In this section we show that for finite alphabets, the truth problem for the 2-variable
fragment $\FO^2(A^*,\subword)$ is decidable. The proof was first
sketched by Kuske~\cite{kuske-private}.

\subsection{Rational relations}
\label{ssec-rat-rel}

We recall the basics of rational relations. See \cite[Chap.\ 3]{berstel79} or
\cite[Chap.\  4]{sakarovitch2009} for more details.

For finite alphabets $A$ and $B$, the \emph{rational relations}
between $A^*$ and $B^*$ are defined as the subsets of $A^* \times B^*$
recognised by asynchronous transducers. The set of rational relations
between $A^*$ and $B^*$ is exactly the closure of the finite subsets
of $A^* \times B^*$ under union, concatenation, and Kleene star.

For example, it is easy to see that the subword relation, seen as a
subset of $A^*\times A^*$ is a rational relation~\cite[Example
  III.5.9]{berstel79}, and that the strict subword relation is
rational too:\footnote{When writing such regular expressions we use
  the vector notation $\bigl[{y\atop x}\bigr]$ to denote $(x,y)$. Note
  that the domain and the range of the relation correspond to the
  bottom and, resp., the top, lines of the vectors. We use $\cdot$ to mean concatenation. }
\begin{xalignat*}{2}
\subword        \;&= \left(\bigcup_{a\in A}\pair{\epsilon}{a}\cup\pair{a}{a}\right)^*
\:,
&
\strictsubword \;&=\;
\subword \cdot \left(\bigcup_{a\in A}\pair{\epsilon}{a}\right)\cdot \subword
\:.
\end{xalignat*}
Define now the \emph{incomparability relation} over $A^*$, denoted
$\perp$, by $u \perp v$ iff $u \not \subword v\land v \not \subword
u$.
\begin{lemma}
The incomparability relation over $A^*$ is a rational relation.
\end{lemma}
\begin{proof}
We cannot simply use the fact that $\not\subword$ and $\not\supword$
are rational
relations since rational relations are not closed under intersection.
The way out is to express incomparability as a \emph{union}
$\perp\ =T_1\cup T_2$ of rational relations, using the following
equivalence
\begin{align}
u\perp v
\;\text{ iff }\;
\overset{(u,v)\in T_1}{\overbrace{(u\not\subword v\land |u|\leq |v|)}}
\lor
\overset{(u,v)\in T_2}{\overbrace{(v\not\subword u\land |v|\leq |u|)}}
\:.
\end{align}
The equivalence holds since $|u|>|v|$ implies
$u\not\subword v$.

We show (see Coro.~\ref{coro-T1-rat}) that $T_1$ is rational. A
symmetric reasoning shows that $T_2$ is rational. This concludes since
the union of two rational relations is rational.
\end{proof}
In the following proof, we write $w(0:-i]$ to denote the prefix of
  length $|w|-i$ of an arbitrary word $w$ (assuming $0\leq i\leq|w|)$.
\begin{lemma}
\label{lem-T1-rat-aux}
$(u,v)\in T_1$ iff there exists an integer $\ell$, a factorisation
  $u=a_1a_2\ldots a_\ell a u'$ of $u$, and a factorisation $v=v_1 a_1
  v_2 a_2 \ldots v_\ell a_\ell b v'$ of $v$ such that\\
---
  $a_1,\ldots,a_\ell\in A$ and $v_1,\ldots,v_\ell\in A^*$ are such that $a_i$ does not occur in $v_i$  for
all  $i=1,\ldots,\ell$,
\\
--- $a,b\in A$ are two
  letters with $a\neq b$, and
\\
--- $u',v'\in A^*$ are two suffixes with
  $|u'|=|v'|$.
\end{lemma}
\begin{proof}
The $(\Longleftarrow)$ direction is clear: the listed conditions
guarantee $|u|\leq|v|$ and $u\not\subword v$.

To see the $(\implies)$ direction, we assume $(u,v)\in T_1$ and write
$u=a_1\ldots a_n$, with $n=|u|$, knowing that $n>0$ since
$u\not\subword v$. We say that $i\in\{0,\ldots,n\}$ is \emph{good} if
$u(0:-i]\subword v(0:-i]$, and \emph{bad} otherwise. Clearly, $n$ is
    good and $0$ is bad. Let $m>0$ be the smallest good index: it is
    easy to check that taking $\ell=n-m$, $a=a_{\ell+1}$ and
    $u'=a_{\ell+2}\ldots a_n$ proves the claim.
\end{proof}
\begin{corollary}
\label{coro-T1-rat}
$T_1$ is a rational relation.
\end{corollary}
\begin{proof}
Lemma~\ref{lem-T1-rat-aux} directly translates as
\[
T_1=\Biggl(\bigcup_{a\in A}\Biggl[\bigcup_{b\neq
    a}\pair{\epsilon}{b}\Biggr]^*\cdot\pair{a}{a}\Biggr)^*
\cdot\Biggl(\bigcup_a\bigcup_{b\neq a}\pair{a}{b}\Biggr)
\cdot\Biggl(\bigcup_{a,a'}\pair{a}{a'}\Biggr)^*
\:.\]
\end{proof}

\subsection{Decidability for $\FO^2$}
\label{ssec-FO2-algo}

Let $\calR \egdef \{=,\ssubword,\ssupword,\perp\}$ consists of the following
four relations on $A^*$: equality, strict subword relation, its
inverse, and incomparability. These four relations form a partition of
$A^*\times A^*$, i.e., for all $u,v \in A^*$, exactly one of $u = v$,
$u \ssubword v$, $u \ssupword v$, and $u \perp v$ holds.

For any $R \in \calR$ and language $L \subseteq A^*$, we define the
\emph{preimage} of $L$ by $R$, denoted $R^{-1}(L)$, as being the language $\{
x \in A^* : \exists y \in L : (x,y)\in R\}$. We saw in
section~\ref{ssec-rat-rel} that each relation $R \in \calR$ is
rational: we deduce that $R^{-1}(L)$ is regular whenever $L$ is.
Furthermore, using standard automata-theoretic techniques, a
description of the preimage $R^{-1}(L)$ can be computed effectively
from a description of $L$.

In the following we consider $\FO^2$ formulae using only $x$ and $y$
as variables. We allow formulae to have regular predicates of the form
$x \in L$ for fixed regular languages $L$ (i.e., we consider the
extended logic). Furthermore, we consider a variant of the logic
where we use the binary relations $\ssubword$, $=$ and $\perp$ instead of
$\subword$. This will be convenient later. The two variants are
equivalent, even when restricting to $\FO^m$ or $\Sigma_m$ fragments: in one direction we observe that $x \subword y$ can be
defined with $x \ssubword y \lor x = y$, in the other direction
one defines $x \ssubword y$ with $x \subword y \land y\not\subword x$ and
$x\perp y$ with $x\not\subword y\land y\not\subword x$. We also use $x
\ssupword y$ as shorthand for $y \ssubword x$.

\begin{lemma}
\label{lem-fo2-reg}
Let $\phi(x)$ be an $\FO^2$ formula with at most one free variable.
Then there exists a regular language $L_\phi\subseteq A^*$ such that
$\phi(x)$ is equivalent to $x\in L_\phi$. Furthermore, a description
for $L_\phi$ can be computed effectively from $\phi$.
\end{lemma}
\begin{proof}
By structural induction on $\phi(x)$. If $\phi(x)$ is an atomic formula of the form $x \in L$, the result is
immediate. If $\phi(x)$ is an  atomic
formula that uses a binary predicate $R$ from $\calR$, the fact that
it has only one
free variable means that $\phi(x)$ is a trivial $x=x$, or $x\ssubword x$, or
\ldots, so that $L_\phi$ is $A^*$ or $\emptyset$.

For compound formulae of the form $\neg \phi'(x)$ or $\phi_1(x) \lor
\phi_2(x)$, we use the induction hypothesis and the fact that regular
languages are closed under boolean operations.

There remains the case where $\phi(x)$ has the form $\exists y\:
\phi'(x,y)$. We first replace any subformulae of $\phi'$ having
the form $\exists x\: \psi(x,y)$ or $\exists y\: \psi(x,y)$ with
equivalent formulae of the form $y \in L_\psi$ or $x \in L_\psi$
respectively, for appropriate languages $L_\psi$, using the induction
hypothesis. Thus we may assume that $\phi'$ is quantifier-free. We now
rewrite $\phi'$ by pushing all negations inside with the following
meaning-preserving transformations:
\begin{xalignat*}{3}
\neg \neg \psi &\to \psi
&
\neg (\psi_1 \lor \psi_2) &\to \neg \psi_1 \land \neg \psi_2
&
\neg (\psi_1 \land \psi_2) &\to \neg \psi_1 \lor \neg \psi_2
\end{xalignat*}
and then eliminating negations completely with:
\begin{xalignat*}{2}
\neg (z \in L) &\to z \in (A^* \setminus L)
&
\neg (z_1\ R_1\ z_2) &\to z_1\ R_2\ z_2 \lor z_1\ R_3\ z_2 \lor z_1\ R_4\ z_2
\end{xalignat*}
where $R_1, R_2, R_3, R_4$ are relations such that $\calR = \{R_1,
R_2, R_3, R_4\}$. Thus, we may now assume that $\phi'$ is a positive boolean
combination of atomic formulae. We write $\phi'$ in disjunctive normal
form, that is, as a disjunction of conjunctions of atomic formulae.
Observing that $\exists y (\phi_1 \lor \phi_2)$ is equivalent to
$\exists y\: \phi_1\lor \exists y\: \phi_2$, we assume w.l.o.g.\  that
$\phi'$ is just a conjunction of atomic formulae. Any atomic formula of
the form $x \in L$, for some $L$, can be moved outside the existential
quantification, since $\exists y (x \in L \land \psi)$ is equivalent
to $x \in L \land \exists y\: \psi$. All atomic formulae of the form
$y \in L$ can be combined into a single one, since regular languages
are closed under intersection.

Finally we may assume that $\phi'(x,y)$ is a conjunction of a single
atomic formula of the form $y \in L$ (if no such formula appears, we
can write $y \in A^*$), and some combination of atomic formulae among
$x \ssubword y$, $x \ssupword y$, $x=y$, and $x \perp y$. If at least
two of these appear, then their conjunction is unsatisfiable, and so
$\phi(x)$ is equivalent to $x\in \emptyset$. If none of them appear,
$\exists y (y \in L)$ is equivalent to $x\in A^*$ (or to
$x\in\emptyset$ if $L$ is empty). If exactly one of them appears, say
$x\  R\  y$, then $\exists y\ (y\in L \land x R y))$ is equivalent to
$x\in L_\phi$ for $L_\phi=R^{-1}(L)$, which is regular as observed
earlier.
\end{proof}

\begin{theorem}
\label{thm-FO2-dec}
The truth problem for $\FO^2(A^*,\subword)$ is decidable.
\end{theorem}
\begin{proof}
Lemma~\ref{lem-fo2-reg} provides a recursive procedure for computing
the set of words that make $\phi(x)$ true. When $\phi$ is a closed
formula, this set is $A^*$ or $\emptyset$ depending on whether $\phi$
is true or not.
\end{proof}

\subsection{Hardness for $\FO^2$}
\label{ssec-FO2-hard}

The main question left open in this paper is the complexity of the
decidable $\FO^2$ theory. The recursive procedure described in
Lemma~\ref{lem-fo2-reg} is potentially non-elementary since nested
negations lead to nested complementations of regular languages.

Our preliminary attempts suggest that the question is difficult. At
the moment we can only demonstrate the following lower bound.

\begin{theorem}
\label{thm-FO2-pspace}
Truth checking for the basic logic, restricting to $\FO^2$ sentences which only use letters (that is, words of length $1$) as constants, is $\PSPACE$-hard.
\end{theorem}
\begin{proof}
We reduce from $\TQBF$, the truth problem for quantified boolean
formulae. W.l.o.g.\  a given instance of $\TQBF$ has the form
$\phi' = \exists p_1 \forall p_2 \ldots \exists p_{2n-1} \forall p_{2n} \phi$.

Consider the alphabet $A$ with $4n$ letters, $T_i$ and $F_i$ for each
$1 \leq i \leq 2n$. A word $w \in A^*$ is intended to encode a
(partial) boolean valuation $\VAL_w$ of the variables $p_1, \ldots,
p_{2n}$: if $T_i$ appears in $w$, $\VAL_w(p_i)=\True$, and if $F_i$
appears in $w$, $\VAL_w(p_i)=\False$. We do not consider
``inconsistent'' words, in which both $T_i$ and $F_i$ appear. Observe
that if $x$ and $y$ represent partial valuations and $x \subword y$,
then $\VAL_y$ extends $\VAL_x$. Conversely, any valuation extending
$\VAL_x$ can be represented by a suitable $y'$ with $x \subword y'$.

 For each $i$, let $\varphi_i(w)$ be a formula
that  says ``the domain of $\VAL_w$ is $\{x_1, \ldots, x_i\}$'':
\[\bigand_{1 \leq j \leq i} ((T_j \subword w \lor F_j \subword w) \land \neg (T_j \subword w \land F_j \subword w)) \land \bigand_{i < j \leq 2n} (T_j \not \subword w \land F_j \not \subword w)\]
We now translate the given $\TQBF$ instance $\phi'$ into an $\FO^2$ sentence $\psi'$ in our logic:
\[
\begin{split}
 \psi' = \exists x (\varphi_1(x) \land \forall y ( (\varphi_2(y) \land x \subword y) \implies \exists x (\varphi_3(x) \land y \subword x \land \ldots \\
\land \exists x (\varphi_{2n-1}(x) \land y \subword x \land \forall y ( (\varphi_{2n}(x) \land x \subword y) \implies \psi )) \ldots )))
\end{split}
\]
where $\psi$ is obtained from $\phi$ by replacing each $p_i$ with $T_i \subword y$.

The formula $\psi'$ uses the two variables $x$ and $y$ alternately, to build up suitable valuations with the appropriate alternation of $\exists$ and $\forall$. It is easy to see that $\phi'$ is true if and only if $\psi'$ is true.

Finally, it was not necessary to assume that $\phi'$ had a strict alternation of $\exists$ and $\forall$, but it makes the presentation of the proof simpler.
\end{proof}

\makeatletter{}\section{Concluding remarks}
\label{sec-concl}

We considered the first-order logic of the subsequence ordering and
investigated decidability and complexity questions. It was known that the $\Sigma_3$ theory is
undecidable and that the $\Sigma_1$ theory is decidable. We settled
the status of the $\Sigma_2$ fragment by showing that it has an
undecidable theory, even when restricting to formulae using no
constants. To remain in the $\Sigma_2$ fragment, our reduction encoded
language-theoretic problems rather than undecidable number-theoretic
logical fragments as is more usual.

We also showed that the $\FO^2$ theory of the subsequence ordering is
decidable using automata-theoretic techniques. The $\FO^2$ fragment is
quite interesting. We note that it encompasses modal logics where the
subsequence ordering correspond to one step (or its reverse) as used
in the verification of unreliable channel systems.

Finally, we provided some new complexity results like
Theorems~\ref{thm-S1-NP} and~\ref{thm-FO2-pspace}. 
\\

We can list a few interesting directions suggested  by this work.
First, on the fundamental side,
the main
question left open is the precise complexity of the $\FO^2$ theory.

Regarding applications, it would be interesting to see how the
decidability results can be extended to slightly richer logics
(perhaps with some extra functions or predicates, or some additional
logical constructs) motivated by specific applications in automated
reasoning or program verification.
\\

\noindent
\textbf{Acknowledgements.}
We thank Dietrich Kuske who outlined  the proof of Theorem~\ref{thm-FO2-dec}.

\bibliographystyle{plain}\bibliography{subwords}

\appendix
\makeatletter{}\section{Proof of Lemma~\ref{lem-simnnp1}}

Assume $|u|=|v|=n+1$ and $u\neq v$ as in the statement of the Lemma.

We say that a word $w$ \emph{distinguishes $u$ and $v$} if $w$ is a
subword of exactly one of $u$ and $v$. We have to prove that there
exists such a distinguisher $w$ with $|w|\leq n$.
\\

Writing a word $w\in A^*$ under the form $w = a_1^{n_1} \ldots
a_k^{n_k}$ where each $a_i$ is a letter so that $a_i \neq a_{i+1}$ for
all $i=1,\ldots,k-1$ and $n_i \geq 1$ for all $i=1,..,k$ is called the
\emph{block factorisation} of $w$. Here $k$ is the number of blocks in
$w$. We now consider several cases:

\begin{itemize}
\item
Assume that  $u$ has
only one block. Then $u = a^{n+1}$ for some $a \in A$, and some
one-letter word distinguishes $u$ and $v$. The same reasoning applies
if $v$ has only one block.

\item
Assume that $u$ and $v$ have at least two blocks each,
and there is some letter $a \in A$ such that $|u|_a\neq |v|_a$.
Then $a^k$ distinguishes $u$ and $v$ for some $k \leq n$.

\item
We are left to deal with cases where $u$ and $v$ have have at least
two blocks, and have the same Parikh image, that is, $|u|_a = |v|_a$
for every $a \in A$.

Assume now that $u$ has exactly
two blocks. Then $u \in a^+b^+$ for some $a,b \in A$ with $a \neq b$.
Since $v$ has the same number of $a$'s and $b$'s but differs from $u$,
we must have $ba \subword v$. But $ba \not \subword u$, so $ba$ is a
distinguisher (here we use the assumption that $n \geq 2$).

\item
Finally assume that $u$ has at least three blocks. Pick a block $B$ of $u$ which is neither the first nor the last, and let $a$ be the unique letter belonging to $B$. Let $\ell=|u|_a$ and write $u$ as $u = s_0 a s_1 a\ldots a s_\ell$. Then
\[|s_0| + \ldots + |s_\ell| = (n+1) - \ell\:.\]
At least two of the numbers $|s_0|, \ldots, |s_\ell|$ are strictly positive, since the two blocks immediately to the left and right of $B$ both exist, and both do not have $a$. Thus for all $i$, $|s_i| < (n+1) - \ell$.

Since
$|v|_a=\ell$, we can write $v = t_0at_1a\ldots at_\ell$. We assume $u
\sim_n v$ and obtain a contradiction. For
each $i$ such that $0 \leq i \leq \ell$, consider the word $z_i = a^i
s_i a^{\ell-i}$. We have $|z_i| \leq n$, and $z_i \subword u$. Since
$u \sim_n v$, we have $z_i \subword v$. Since both $z_i$ and $v$ have
exactly $\ell$ occurrences of $a$, we have $s_i \subword t_i$. This
holds for all $i$, so $u \subword v$. But $|u| = |v|$, so $u = v$,
which is a contradiction.
\end{itemize}

\end{document}